\documentclass[runningheads]{llncs}
\usepackage[T1]{fontenc}
\usepackage{tikz}
\usetikzlibrary{arrows,automata}
\usepackage{tzplot}
\usepackage{color}
\usepackage{amsmath,amssymb,amsfonts,stmaryrd,extarrows,dsfont}
\usepackage{graphicx}
\usepackage{complexity}
\usepackage{graphicx}
\usepackage{todonotes}
\usepackage{bbold}
\usepackage{subfigure}
\usepackage{url}
\usepackage{enumerate}
\usepackage{times}
\usepackage{hyperref}
\hypersetup{
    colorlinks=true,
    linkcolor=black,
    filecolor=black,      
    urlcolor=black,
    citecolor=black
    }

\newcommand{\flow}[0]{{\sf flow}}

\newcommand{\labeling}{{\sf lbl}}
\newcommand{\rel}{{\sf R}}

\newcommand{\playerone}{\ensuremath{\mathsf{Pl_1}}}
\newcommand{\playertwo}{\ensuremath{\mathsf{Pl_2}}}
\newcommand{\player}[1]{\ensuremath{\mathsf{Pl_{#1}}}}
\newcommand{\hist}[1]{\ensuremath{\mathsf{Hist_{#1}}}}
\newcommand{\Rset}{\mathbb{R}}

\newcommand{\obs}{{\sf Obs}}

\newcommand{\Constr}{{\sf Cnstr}}
\newcommand{\post}{{\sf rst}}
\newcommand{\game}[0]{\mathcal{G}}
\newcommand{\mvs}[1]{{\sf Mvs(#1)}}

\newcommand{\initSingame}{ISR~game }
\newcommand{\act}{{\sf Act}}
\newcommand{\var}{{\sf X}}
\newcommand{\edge}{{\sf E}}
\newcommand{\runs}{{\sf Runs}}
\newcommand{\tr}{{\sf tr}}
\newcommand{\outcome}{{\sf Out}}

\begin{document}
\title{Deciding the synthesis problem for hybrid games through bisimulation}
%
%
\author{Catalin Dima\inst{1}\and
Mariem Hammami\inst{1}\and
Youssouf Oualhadj\inst{1,2}\and
Régine Laleau\inst{1}}
\authorrunning{C. Dima, M. Hammami, Y. Oualhadj, R. Laleau}
%
\institute{LACL, Université Paris-Est Créteil, F-94010 Créteil, France \and CNRS, ReLaX, IRL 2000, Siruseri, India\\
\email{firstname.lastname@u-pec.fr}}

\maketitle              
\begin{abstract}
Hybrid games are games played on a finite graph endowed with real variables which may model behaviors of discrete controllers of continuous systems.  
The synthesis problem for hybrid games is decidable for classical objectives (like LTL formulas)
when the games are initialized singular, meaning that the slopes of the continuous variables are piecewise constant and variables are reset whenever their slope changes. 
The known proof adapts the region construction from timed games. 

In this paper we show that initialized singular games can be reduced, via a sequence of alternating bisimulations, to timed games, generalizing the known reductions by bisimulation from 
initialized singular automata to timed automata.
Alternating bisimulation is the generalization of bisimulation to games, accomodating a strategy translation lemma by which, 
when two games are bisimilar and carry the same observations, each strategy in one of the games can be translated to a strategy in the second game 
such that all the outcomes of the second strategy satisfies the same property that are satisfied by the first strategy. 
The advantage of the proposed approach is that one may then use realizability tools for timed games to synthesize a winning strategy for a given objective, and then use the strategy translation lemma 
to obtain a winning strategy in the hybrid game for the same objective. 
\keywords{Controller synthesis \and Alternating bisimulation \and Hybrid games.}
\end{abstract}

\nocite{AlurD94,FaellaTM02,BehrmannCDFLL07}

\begin{credits}
\subsubsection{\ackname}
This work was funded by grant ANR-17-CE25-0005 
the \href{https://discont.loria.fr}{DISCONT} Project
 from the \emph{Agence Nationale
de la Recherche (ANR)}.

The first and third authors are supported by grant ANR-20-CE25-0012
the \href{https://www.irif.fr/users/maveriq}{MAVeriQ} Project
 from the \emph{Agence Nationale
de la Recherche (ANR)}.
\end{credits}
\section{Introduction}
\label{sec:typesetting-summary}
In order to describe cyber-physical systems in which discrete and continuous physical 
processes interact with each other, many mathematical modelling techniques have evolved as 
a tool. The most common model known so far is the one of hybrid automata~\cite{Henzinger2000}. 

\textbf{Controller synthesis and game theory} A reactive system interacts with its environment, when this system displays 
time-sensitive behaviors, it is described by hybrid games~\cite{HenzingerHM99}.
In practice, some behaviors depend on external 
factors like weather or temperature. These behaviors are by definition uncontrollable, hence for the system behaves as desired,  one needs to find ways to limit the effect of these behaviors. 
We model this situation through the game theoretic metaphor.
We suppose a game played between two players; the first player $\playerone$ models the system and is supposed to be controllable, and the second player $\playertwo$  models the environment and is supposed to be antagonistic.

The specification to ensure is modeled as an objective (\textit{i.e., } a subset of the possible executions of the system) that $\playerone$ has to enforce against any behavior of $\playertwo$. This behavior is formalized by the notion of a strategy. Therefore, the \emph{control policy} that we aim to implement is any strategy that ensures the desired specification.

The control problem for hybrid games consists in answering the question of whether there exists a  strategy of the controller that ensures a given specification. This problem is undecidable for  hybrid games, but decidable 
for initialized singular games \cite{HenzingerHM99}, which are games in which the first derivative of each continuous variable is piecewise constant, and whenever this derivative changes, the variable must be reset to some rational value. 
In this paper we investigate the existence of decidable subclasses for this problem.
The original proof from \cite{HenzingerHM99} shows that the region construction for timed automata \cite{AlurD94} can be adapted for initialized singular games. 
Implementing this result requires translating the whole machinery for deciding timed games to the case of initialized singular games, namely the notion of zones which generalize regions, and the controlled predecessor operators on zones. 
In this article, we show that initialized singular games can be reduced to timed games \cite{FaellaTM02}. 
Such a reduction would allow a more direct application of existing tools for solving timed games like UppAal TIGA \cite{BehrmannCDFLL07}. 
The reduction is obtained using the notion of \emph{alternating simulation}, which, roughly speaking, is a simulation relation that preserves winning strategies,
initially defined for discrete game structures \cite{alur-alternating-refinement} and adapted to concurrent hybrid games in \cite{HENZINGER199894}. 

\textbf{Contribution} 
Inspired by the results of \cite{HENZINGER199894} over rectangular hybrid automata, 
we build, from each initialized singular game, a timed game which is alternating bisimilar with the original game. 
By generalizing \cite{HENZINGER199894}, the construction passes through intermediate steps in which 
the original hybrid game is transformed into a stopwatch game, then into an updatable timed game, and finally into a timed game. 
We also show that each construction comes in pair with an \emph{alternating bisimulation} proving that the resulting game is bisimilar with the previous one. 
We note that, in some cases, the bisimulation is not a bijection.

\textbf{Organization} We start by giving the necessary notions and proper definitions concerning hybrid games and turn based games. 
We then proceed to the construction of a timed turn based game from a given initialized singular turn based game with a bisimulation relation between the two games. 
This construction is achieved, as mentioned above, through three successive transformations, where in the intermediary steps we build stopwatch games, resp. updatable timed games. 
Each transformation is accompanied by a relation between the sets of configurations of the two games, which is shown to satisfy the properties of alternating bisimulations. 
We end with a short section with conclusions and future work. 

\section{Turn-based Hybrid Games}

Given a set of variables $\var$, the set of simple compact constraints over $\var$ is defined by the following grammar: 
\[
\varphi := x \in I \mid \varphi \wedge \varphi 
\]
where $I$ is a compact interval with rational bounds and $x$ is a variable in $\var$.
For a constraint $\varphi$ which contains the conjunct $x \in [a,b]$, $\varphi(x)$ denotes the interval $[a,b]$. We consider that each constraint contains a single conjunct of the form $x \in I$ for each variable $x$.

We extend the framework of Initialized Singular Automata from \cite{HENZINGER199894} to the game setting. Informally, an initialized singular game is played between two players,
called \playerone{} and \playertwo.
The set of locations is partitioned into two disjoint sets, the first one is controlled by \playerone, and the second by \playertwo. A play is obtained by the following interaction, the players take turns by proposing moves as follows: when a player is in a location that they control.
This player proposes an action and a timed delay, hence a new configuration is obtained. The game proceeds from this fresh configuration. By repeating this interaction forever, the two players produce an infinite run for which one of the player wants to satisfy some property (usually called the \emph{objective} of the game).

\begin{definition}[Initialized Singular Game]
An initialized singular game (\initSingame for short) denoted $\game_{S}$ of dimension $n$ between two players $\playerone$ and $\playertwo$ is a tuple
$\game_{S} = (L_1, L_2, l_0, \var, \act, \obs, \flow, \edge, \labeling)$ such that: 
\begin{itemize}
  \item $L_1$ is the finite set of locations that belong to $\playerone$ and $L_2$ is the finite set of locations belonging to $\playertwo$, with $L_1\cap L_2 =\emptyset$. 
  \item $l_0 \in L_1$ is the initial location, assumed to belong to $\playerone$.    
  \item $\var$ is a set of $n$ real variables.
  \item $\act$ is a finite set of actions.
  \item $\obs$ is a finite set of observations. 
  \item $\labeling: (L_1 \cup L_2) \to \obs$ is the labeling function that labels each location with an observation. 
  \item $\flow: (L_1 \cup L_2) \times \var \to \mathbb Q$ is the value of the derivative which constrains the evolution of each variable in each location.
  \item $\edge \subseteq \left(L_1 \cup L_2\right) \times \act \times \Constr(\var) \times \post(\var)\times \left(L_1 \cup L_2\right) 
  $ 
  where $\Constr(\var)$ is the set of all simple compact constraints over $\var$ and $\post(\var) : \var \to \mathbb Q \cup \{ \bot\}$ is the set of functions that we call \emph{reset function}.
  We will denote an edge $e \in \edge$ as a tuple
  $e=(l,a,\varphi_e, \post_e,l')$ such that:
  \begin{itemize}
  \item $\varphi_e \in \Constr(\var)$, in order to take the edge $e$ the value of the variables $\var$ must satisfy the constraint $\varphi_e$. 
  \item $\post_e$ is a reset function of the variables when the edge $e$ is taken, $\post_e: \var \to \mathbb Q \cup \{ \bot\}$, $\post_e(x) \in \mathbb{Q}$ means that $x$ is reset to a new value and when $\post_e(x) = \bot$ it means that $x$ is not reset for any $x\in \var$. By abuse of notation, for a variable valuation $v : \var \rightarrow \Rset$, we also denote $\post_e(v)$ the valuation defined by 
  \[
\post_e(v)(x) =
  \begin{cases}
      v(x) & \text{ when } \post_e(x)  = \bot\\
      \post_e(x) & \text{ otherwise} 
  \end{cases}
  \]
  \item when $\flow(l,x) \neq  \flow(l',x)$ then $\post_e(x) \neq \bot$. In other words, whenever a variable $x \in \var$ changes its dynamics then its value is 
  reset to $\post_e(x)$.
  \end{itemize}
\end{itemize}
\end{definition}
The semantics of an \initSingame is defined as a transition system 
$$T(\game_S) = (Q(\game_S), \mvs{\game_S}, q_0,\delta_S)$$
where $Q(\game_S)$ is the set of \emph{configurations}, i.e., all the couples
$(l,v) \in (L_1 \cup L_2) \times \mathbb{R}^{n}$ with $q_0  = (l_0,\Vec{0})$, 
$\mvs{\game_S} = \act \times \Rset_+$ is the set of moves, and $\delta_S$ consists of transitions of the form $(l,v) \xrightarrow{(a,t)} (l',v')$
for which there exists an edge $e=(l,a, \varphi_e,\post_e,l')$ such that, for each variable $x \in \var$:
\begin{itemize}
    \item $v(x) + t \cdot \flow(l,x) \in  \varphi_e(x) $. 
    We will also use the vectorial notation $v'= v+t\cdot \flow(t,\cdot)$
    where $\flow(l,\cdot)$ denotes the vector of variable flows.
    \item $v'(x) = \post_e(x)$ whenever $\post_e(x) \neq \bot$. 
    \item $v'(x) = v(x) +t \cdot \flow(l,x)$ whenever $\post_e(x) = \bot$.    
\end{itemize}
We denote $Q_1$ the set of configurations whose location belongs to $\playerone$, i.e. $Q_1(\game_S) = \{(l,v) \in Q(\game_S) \mid l \in L_1 \}$ and, similarly, $Q_2(\game_S)= \{(l,v) \in Q(\game_S) \mid l \in L_2 \}$. 
We extend the function $\labeling$ over the set of all the configurations $Q_1 \cup Q_2$ as expected.
A $\emph{run}$ $\rho$ in an $\initSingame$ is a finite or infinite alternating sequence of configurations and moves, 
$\rho = q_0 \xrightarrow{(a_1,t_1)} q_1 \xrightarrow{(a_2,t_2)} q_2 \ldots$.
The set of runs in $T(\game_S)$ is denoted $\runs(\game_S)$. A history in $T(\game_S)$ is any prefix of a run in $\runs(\game_S)$, we denote by $\hist{i}(\game_S)$ the set of histories ending in a configuration of $\player{i}$.
A \emph{trace} is a finite or infinite sequence of observations, and the trace induced by a run $\rho$ is simply $\tr(\rho) = \labeling(q_0)\cdot \labeling(q_1) \cdot \labeling(q_2)\ldots$
\begin{definition}[Strategy on $T(\game_S)$]
Let $i=1,2$. A strategy $\sigma_i$ of $\player{i}$ is a function
that maps finite sequences of configurations ending in a configuration of $\player{i}$ to a move (an action and a time delay). 
\end{definition}
Given two strategies $\sigma_1$ and $\sigma_2$, their \emph{outcome} $r = q_0\xrightarrow{m_1} q_1 \xrightarrow{m_2} \ldots \in \runs(\game_S)$ is obtained as follows:
\begin{align*}
    q_1 &= \delta_S(q_0,\sigma_1(q_0)) \text{ and } m_1 = \sigma_1(q_0) \\
    q_{i+1} &= 
    \begin{cases}
        \delta_S(q_i,\sigma_1(q_0\ldots q_i)) & \text{ if } q_i \in Q_1(\game_S)
        \text{ and } m_i = \sigma_1(q_0\ldots q_i) \\
        \delta_S(q_i,\sigma_2(q_0\ldots q_i)) & \text{ if } q_i \in Q_2(\game_S) 
        \text{ and } m_i = \sigma_2(q_0\ldots q_i) 
    \end{cases}
\end{align*}
We denote $\outcome(\sigma_1)$ the set of all runs that are induced by the strategy 
$\sigma_1$ of $\playerone$ and some strategy of $\playertwo$.

The study of games involves the construction of strategies whose set of outcomes satisfies some desired property. 

Semantically, a \emph{property} is a subset of the set of infinite traces. We then say that a strategy $\sigma_1$ for \playerone{} satisfies a property $P$ if the set of infinite traces which correspond to outcomes of $\sigma_1$ is included in $P$.

It is known that the problem of synthesizing strategies for properties defined by Linear Temporal Logic formulas is decidable for the case of
initialized singular games, as proved in \cite{HenzingerHM99} (Theorem 3.1) where an adaptation of the region construction is used.
In the rest of the paper, we focus on an alternative approach, namely reducing the synthesis problem for ISR-games to the synthesis problem for 
timed games by means of successive \emph{alternating bisimulation} reductions. We start by recalling the notion of alternating bisimulation from \cite{HenzingerHM99}.

\subsection{Alternating simulation for turn-based hybrid games}
In this section we adapt the notion of alternating simulation for the case of (turn-based) ISR-games
and a semantic version of the strategy translation property from \cite{alur-alternating-refinement,HenzingerHM99}.

\begin{definition}[Alternating Simulation on Initialized Compact singular Game]\label{def-bisim}
Let $\game^1_S$ and $\game^2_S$ be two initialized singular games with the same set of observations $\obs$, and let
$T(\game^1_S) = (Q(\game^1_S), \mvs{\game^1_S}, q^1_{0},\delta^1_S)$ and
$T(\game^2_S) = (Q(\game^2_S), \mvs{\game^2_S}, q^2_0,\delta^2_S)$ 
be respectively the induced transition systems.\\
Let ${\sf R}\subseteq (Q_1(\game^1_S) \times Q_1(\game^2_S)  ) \cup (Q_2(\game^1_S) \times Q_2(\game^2_S) )$ be a binary relation.
The relation ${\sf R}$ is called a \textbf{simulation} if, 
for any two configurations $p=(l,v) \in Q(\game^1_S) $ and 
$q=(s,y) \in Q(\game^2_S)$, 
if $(p, q) \in {\sf R}$ then:
\begin{enumerate}
\item $\labeling^1(p)=\labeling^2(q)$
\item If $p\in Q_1(\game^1_S)$ and $q \in Q_1(\game^2_S)$, then for all $ (m,t) \in \mvs{\game^1_S} $ with $p'=(l',v')= \delta^1_S((l,v),(m,t))$ there exists $(m',t') \in \mvs{\game^2_S} $ and $q'=(s',y')=\delta^2_S((s,y),(m',t'))$ such that $(p', q') \in \rel$.
\item If $p\in Q_2(\game^1_S) $ and $q \in Q_2(\game^2_S) $, then for all $ (m',t') \in \mvs{\game^2_S}$ and $q'=(s',y')=\delta^2_S(q,(m',t'))$ there 
exists $(m,t) \in \mvs{\game^1_S}$ and $p'=(l',v')= \delta^1_S(p,(m,t))$ such that $ (p,' q') \in \rel$.
\end{enumerate}
\end{definition}
We say that $T(\game^2_S)$ simulates $T(\game^1_S)$ and we denote 
$T(\game^1_S) \preceq_s T(\game^2_S)$ if there exists a simulation relation $\rel \subseteq (Q_1(\game^1_S) \times Q_1(\game^2_S)  )\cup (Q_2(\game^1_S) \times Q_2(\game^2_S) )$
which satisfies all of the above conditions and contains the pair of initial configurations, i.e. 
$(q^1_{0}, q^2_0) \in \rel$. By extension we also denote sometimes $\game^1_S \preceq_s \game^2_S$

Moreover, $T(\game^1_S)$ and $T(\game^2_S)$ are bisimilar if there exists simulation relation $\rel$ witnessing $T(\game^1_S) \preceq_s T(\game^2_S)$ and such that $\rel^{-1}$ witnesses $T(\game^2_S) \preceq_s T(\game^1_S)$.
\begin{lemma} [Simulation Composition] \label{compo}
Let $\game^1$, $\game^2$ and $\game^3$ be three turn based game structures 
and their respective semantics $T(\game^1), T(\game^2), T(\game^3)$.
Assume that $\alpha$ is a simulation between $T(\game^1)$ and $ T(\game^2)$, $\beta$ a simulation between $T(\game^2)$ and $ T(\game^3)$ hence $\gamma = \alpha \circ \beta$ is a simulation between $T(\game^1) $ and $ T(\game^3)$.
\end{lemma}

The proof of the above lemma is straightforward and is left as a simple exercice.

\begin{lemma}\label{lemme-transfert-strategies}
If $T(\game^1_S) \preceq_s T(\game^2_S)$ then for all $\sigma_1$ strategy of $\playerone$ in $T(\game_S^1)$, there exists $\sigma_1'$ a strategy of $\playerone$ in $T(\game^2_S)$ such that 
\[
\{ \tr(\rho') \mid \rho' \in \outcome({\sigma_1'}) \}  \subseteq
\{ \tr(\rho) \mid \rho \in \outcome({\sigma_1}) \}
\]
\end{lemma}
The relevance of this lemma is the following: if we may prove that a game $\game$ is 
bisimilar with a simpler game $\game'$ and are able to build a strategy $\sigma'$ for \playerone in $\game'$ which 
satisfies some property, then Lemma \ref{lemme-transfert-strategies} allows us to translate $\sigma'$ into a strategy $\sigma$ which 
satisfies the same property. Therefore, if $\game'$ lies in a decidable class of games $\Gamma$, we may use the decision algorithm for $\Gamma$ to 
synthesize $\sigma'$ and, hence, provide a methodology for attacking the strategy synthesis problem in the class of games where $\game$ lies. 

The plan for the rest of the paper is then to show that ISR games are alternating bisimilar with simpler games in which 
the strategy synthesis problem is decidable.

\section{From Initialized Singular Games to Stopwatch Game}
We extend the construction from \cite{HENZINGER199894} used for automata to the more general case of turn-based games. The first step is 
to reduce, through an alternating bisimulation, each initialized singular game to a 
\emph{stopwatch game}. 
\begin{definition}[Initialized Stopwatch Game]
An initialized stopwatch game 
of $n$-dimension is an initialized singular game in which the flow is either $1$ or $0$, that is, for any location $l$ and variable $x$, $flow(l,x) \in \{0,1\}$.
\end{definition}

\subsection{Transformation}
\label{trans:singToStop}
Given an initialized singular game 
$\game_{S} = (l_0, L_1, L_2, \var, \act, \obs, \flow, \edge, \labeling)$
we construct from $\game_S$ an initialized stopwatch game $\game_W$, where we change the dynamic to $\flow(l,x)=1$ when it is not zero, 
for that we adapt the constraints and the reset functions; we divide the constraints values with the flow of the current location and the reset values with the flow of the successor location. The game is built as follows: 
$$\game_{W} = (l_0, L_1, L_2, \var, \act, \obs, \flow_W, \edge_W, \labeling)$$
\begin{enumerate}
  \item $l_0$, $L_1$, $L_2$, $\act$, $\obs$ and $\labeling$ are kept the same as $\game_S$.
  \item $\flow_W: L \times \var \to \mathbb{Q}^n$ where for all $x \in \var$, 
  \begin{align}
  \label{eq:flow}
  \flow_W(l,x)=
     \begin{cases}
      0 & \text{ if }\flow(l,x)=0\enspace,\\
      1 & \text{otherwise}\enspace.
          \end{cases}
    \end{align}
   \item $\edge_W$ is the set of edges $e_W=(l,a,\varphi_{e_W}, \post_{e_W},l')$ such that 
   \begin{enumerate}
   \item there exists $e=(l,a,\varphi_e,\post_e,l') \in \edge$.
       \item If we denote 
         \begin{align*}
 \varphi_{e_W}(x)=
     \begin{cases}
      \frac{\varphi(x)}{\flow(l,x)} & \text{ if }\flow(l,x)\neq 0\enspace,\\
      \varphi(x) & \text{otherwise}\enspace.
          \end{cases}
    \end{align*}
then the constraint $\varphi_{e_W}$ is the following conjunction: $ \varphi_{e_W}= \bigwedge_{x\in \var} \varphi_e(x) $. 
       \item $\post_{e_W}: \var \to \mathbb{Q}^n \cup \{ \bot\}$ where 
  to each $x \in \var$ and edge $e=(l,a,\varphi_e,\post_e,l')\in \edge$
        \begin{align}
        \label{eq:post}
       \post_{e_W}(x) = 
        \begin{cases}
            \bot & \text{ if }\post(x)= \bot\enspace,\\
            \frac{\post(x)}{\flow(l',x)} & \text{ if } \flow(l',x)\neq 0  \enspace,\\
            \post(x) & \text{ if } \flow(l',x) = 0 \enspace.
        \end{cases}
    \end{align}
   \end{enumerate}
\end{enumerate}
For the sequel, we denote the semantics of $\game_S$ as 
$T(\game_S) = (Q(\game_S), \mvs{\game_S}, q_0,\delta)$
and the semantics of $\game_{W}$ as 
$T(\game_W) = (Q(\game_W), \mvs{\game_W}, q_0,\delta_W)$.

For the sequel, given a configuration $(l,v)\in Q(\game_S)$, we denote $v^*$ the following variable valuation:
\begin{align}
\label{eq:vstar}
\forall~x \in \var\colon~
v^*(x) = 
\begin{cases}
   \frac{v(x)}{\flow(l,x)}  & \text{ if } \flow(l,x)\neq 0 \\
   v(x) & \text{ otherwise }
\end{cases}
\end{align}

We define the mapping
$\gamma_1: Q(\game_S) \to Q(\game_W)$ as follows: 
\[\forall q=(l,v) \in Q(\game_S),~\gamma_1 \colon (l,v) \mapsto (l,v^*)\enspace.\]
Clearly enough, $\gamma_1^{-1}$, the inverse of $\gamma_1$, exists.
Therefore, for any $(l,v^*)\in Q(\game_W)$, 
$\gamma_1^{-1} \colon (l,v^*) \mapsto (l,v)$ where 
\begin{align}
\forall~x \in \var\colon
v(x) = 
\begin{cases}
   v^*(x) \cdot {\flow(l,x)}  & \text{ if } \flow(l,x)\neq 0 \enspace,\\
   v^*(x) & \text{ otherwise}\enspace.
\end{cases}
\end{align}
By construction of $T(\game_W)$ we have $(l,v)\in Q(\game_S)$.
\subsection{Bisimulation between $T(\game_S)$ and $T(\game_W)$}
We now show that the mapping $\gamma_1$ defined above is an alternating bisimulation.
\begin{lemma}
$\gamma_1$ witnesses that $T(\game_S) \preceq_s T(\game_W)$. 
\end{lemma}
\begin{proof}
Note that, by definition, $q_0^*=\gamma_1(q_0) = q_0$ the initial configurations in the two games $(q_0,q^*_0)\in \gamma_1$ as required. 

Let $q =(l,v)$ be a configuration in $Q(\game_S)$, and let $q^* = \gamma_1(q)=(l,v^*)$ be a configuration in $Q(\game_W)$. 
Since the location $l$ is the same, the observation is preserved thus $\labeling(q) = \labeling(q^*)$, so the first point in Def.~\ref{def-bisim} holds.

We need to show that the 2nd item in Def.~\ref{def-bisim} holds, that is, 
for any two configurations $(l,v)$ in $Q_1(\game_S)$, and $(l,v^*)$ in 
$Q_1(\game_W)$ with $((l,v),(l,v^*)) \in \gamma_1 $
and for any move $m$ in $\mvs{\game_S}$, there exists a move $m^*$ in $\mvs{\game_W}$
such that the following holds:
\begin{align}
\label{eq:singSim1}
    \left(
\delta_S((l,v),m),
\delta_W((l,v^*), m^*)
\right) \in \gamma_1
\enspace.
\end{align}
The trick here, is to use the same move $m$ in both $T(\game_S)$ and $T(\game_W)$. 

Let $(l,v)$ be a configuration of $\playerone$, and let $m$ be a move, we first prove that $\delta_W((l,v^*), m)$ is well defined.
Since $(l',v') = \delta_S((l,v),m)$ with $m=(a,t)$, 
there exists $e=(l,a,\varphi_e,\post_e,l')$ in $\game_S$ 
such that for each $x \in \var$,
\begin{equation}\label{eq-tr3-v1}
v(x) + t \cdot \flow(l,x) \in \varphi_e(x)
\end{equation}
and $v' = \post_e(v)$.
By transformation of Sec.~\ref{trans:singToStop}, there exists an edge
$e^*=(l,a,\varphi_{e^*},\post_{e^*},l^{'})$ in $\game_W$.
Now, we show that:
$v^{*}+ t \cdot \flow_W(l,\cdot )$ satisfies $\varphi_{e^*}$.
To this end, observe that Identity \eqref{eq-tr3-v1} implies that for all $ x \in \var$: 
\begin{align*}
 \flow(l,x) \neq 0 &\implies
 \frac{v(x)}{\flow(l,x)} + t \in \frac{\varphi_e(x)}{\flow(l,x)}\\
 &\implies v^*(x) + t \in  \varphi_{e^*}(x) 
\end{align*}
Where the second implication is because $\flow(l,x) \neq 0 \implies \flow_W(l,x) = 1$.

When $\flow(l,x) = 0$, $v^*(x) = v(x)$, $\varphi_{e^*}(x) = \varphi_e(x)$ and 
Identity~\eqref{eq-tr3-v1} becomes $v(x) \in \varphi_e(x)$. 
Hence $v^*(x) \in \varphi_{e^*}(x)$.
This implies that 
$\delta_W(((l,v^*),m)$ is indeed well defined. 
We still need to show that it is the image of $\delta_S((l,v),m)$ by $\gamma_1$.
 
Let $(l',v') = \delta_S((l,v),m)$, then we need to study the following cases: \\
(1) For $x\in \var$ with $\flow(l, x) \neq 0$ and $\post_e(x) = \bot$:
\begin{align*}
  \gamma_1((l',v'(x)) &= 
  \left(l', \frac{v'(x)}{\flow(l',x)}\right)
  =\left(
  l', \frac {v(x) + t \cdot \flow(l,x)}{\flow(l',x)}
  \right)\\
  &= \left(
  l',\frac{v(x)}{\flow(l,x)} + t
  \right)\text{ because $\flow(l',x)=\flow(l,x)$, cf. Eq.~\eqref{eq:vstar},}\\
  &=(l', v^*(x) + t)\text{ because $\flow_W(l',x)=1$, cf. Eq.~\eqref{eq:flow}}\\
  &= \delta_W((l,v^*(x)),m) \enspace.
  \end{align*}
(2) For all $x\in \var$ with $\flow(l, x) = 0$ and $\post_e(x) = \bot$:
\begin{align*}
  \gamma_1((l',v'(x)) &= 
  \gamma_1(l', v(x)) \text{ because $\flow(l,x) = 0$, cf. Eq.~\eqref{eq:flow}}\\ 
  &=\left( l', v^*(x) \right)= \delta_W((l,v^*(x)),m) 
\text{ because $\flow_W(l',x) = 0$, cf. Eq.~\eqref{eq:flow}.}
  \end{align*}
(3) For all $x\in \var$ with $\flow(l', x) \neq 0$ and $\post_e(x) \neq \bot$:
\begin{align*}
  \gamma_1((l',v'(x)) &= 
  \left(l', \frac{v'(x)}{\flow(l',x)}\right)\\ 
  &=\left(
  l', \frac {\post_e(x)}{\flow(l',x)}
  \right) \text{ because $\post_e(x) \neq \bot$}\\
  &=\left(l',\post_{e^*}(x)\right) \text{ because $\flow(l',x) = 1$, cf. Eq.~\eqref{eq:flow}}\\
  &= \delta_W((l,v^*(x)),m) \enspace.
  \end{align*}
(4) Finally for all $x\in \var$ with $\flow(l, x) = 0$ and $\post_e(x) \neq \bot$:
\begin{align*}
  \gamma_1((l',v'(x)) &= 
  \left(l',v'(x)\right)
  =\left(
  l', \post_e(x)
  \right)
  =\left(l',\post_{e^*}(x)\right) 
  = \delta_W((l,v^*(x)),m) \enspace.
  \end{align*}
In all of the above cases, we have shown that $\gamma_1(\delta_S((l,v), m)) = \delta_W(\gamma_1((l,v)), m)$,
that is, $\gamma_1$ preserves the transition relation in both 
$T(\game_S)$ and $T(\game_W)$.

Now, for proving the 3rd point in Def. \ref{def-bisim},
Let $q=(l,v)$ be a configuration of $\playertwo$, and let $m$ be a move in $\mvs{\game_W}$, we take the same move in $\mvs{\game_S}$, we first prove that $\delta_S((l,v),m)$ is well defined in $T(\game_S)$. \\
Since $(l',{v^*}')= \delta_W((l,v^*),m)$ with $m=(a,t)$ there exists $e^*=(l,a,\varphi_{e^*}, \post_{e^*},l')$ in $\game_W$ such that for each $x\in \var$, 
\begin{equation}\label{eq-tr3-v2}
v^*(x) + t \cdot \flow_W(l,x) \in \varphi_{e^*}(x)
\end{equation}
and ${v^*}'= \post_{e^*}(v)$. By transformation of Sec.~\ref{trans:singToStop}, the existence of $e^*$ in $\game_W$ is the result of the existence of $e=(l,a,\varphi_e, \post_e,l')$ in $\game_S$. Now we show that $v+ t \cdot \flow(l,\cdot)$ satisfies $\varphi_e$. Observe that identity \eqref{eq-tr3-v2} implies that for all $x \in \var$:
\begin{align*}
 \flow_W(x) \neq 0 &\implies v^*(x)+ t \in \varphi_{e^*}(x) \\
& \implies \frac{v(x)}{\flow(l,x)} + t \in \frac{\varphi_e(x)}{\flow(l,x)}\\
 &\implies v(x) + t \cdot \flow(l,x) \in \varphi_e(x)
\end{align*}
Where the second implication is because $\flow_W(x)=1$ implies $\flow(x) \neq 0$. 
When $\flow_W(l,x)=0$, $v(x)=v^*(x)$, $\varphi_e(x)= \varphi_{e^*}(x)$, and identity \eqref{eq-tr3-v2} becomes $v^*(x)\in \varphi_{e^*}(x)$, hence $v(x)\in \varphi_{e}(x)$. This implies that $q'=\delta_S((l,v),m)$ is indeed well defined with $q'=(l',v')$ and $v'=\post_e(v+t\cdot \flow(l,\cdot))$. We still need to prove that the image of $\delta_S((l,v),m)$ by $\gamma_1$ is $\delta_W((l,v^*),m)$.

For $(l',{v^*}')= \delta_W((l,v^*),m)$ and we name $(l',v')= \delta_S((l,v),m)$, we have $\forall x \in \var$ with $\flow_W(l,x) \neq 0$ and $\post_{e^*}(x)= \bot$:
\begin{align*}
 (l',{v^*}'(x)) &=(l',{v^*}(x) +t ) =\left(l',\frac {v(x)}{\flow(l,x)} + t \right)\\
 &=\left(l',\frac {v(x) +t \cdot \flow(l,x) } {\flow(l,x)}\right) 
            =\left(l',\frac { v(x) +t \cdot \flow(l,x) } {\flow(l',x)} \right)\\
 &=\left(l',\frac {v'(x)} {\flow(l',x)} \right) 
       = \gamma_1(l', v'(x))
    \end{align*}
Consider the case $\flow_W(l,x) = 0$, hence for all $x \in \var$ with $\post_{e^*}(x)= \bot$:
\begin{align*}
 (l',{v^*}'(x)) &=(l',{v^*}(x)) =\gamma_1(l',v(x)) = \gamma_1(l', v'(x))
\end{align*}
Now for all $x \in \var$ with $\flow_W(l',x) \neq 0$ and $\post_{e^*}(x) \neq \bot$: 
\begin{align*}
 (l',{v^*}'(x)) &=(l',\post_{e^*}(x) ) 
    =\left(l',\frac {\post_{e^*}(x)}{\flow(l',x)}\right)\\
  &=\left(l',\frac {v(x) +t \cdot \flow(l,x) } {\flow(l,x)}\right)
  =\left(l',\frac { v(x) +t \cdot \flow(l,x) } {\flow(l',x)} \right)\\
&=\left(l',\frac {v'(x)} {\flow(l',x)} \right)
     = \gamma_1(l', v'(x))
    \end{align*}
Finally, for all $x \in \var$ with $\flow_W(l',x) = 0$ and $\post_{e^*}(x) \neq \bot$: 
\begin{align*}
 (l',{v^*}'(x)) =(l',\post_{e^*}(x) ) 
 =(l',\post_{e}(x)
=(l',v'(x))
= \gamma_1(l', v'(x))
    \end{align*}
Therefore,
$\gamma_1(\delta_S((l,v), m)) = \delta_W(\gamma_1((l,v)), m)$.
This ends the proof of this lemma. 
\qed\end{proof}

\begin{lemma}
$\gamma^{-1}$ witnesses that
$T(\game_W) \preceq_s T(\game_S)$.  
\end{lemma}

\begin{proof}
Note that, by definition, $q_0=\gamma_1^{-1}(q_0^*)$ the initial configurations in the two games as required. 

Let $q^*=(l,v^*)$ be a configuration in $Q(\game_W)$ and let $q=(l,v)= \gamma_1^{-1}(q^*)$ be a configuration in $Q(\game_S)$.
Since the location $l$ is the same, the observation is preserved thus $\labeling(q) = \labeling(q^*)$, so the first point in Def.~\ref{def-bisim} holds.

We now show that $\gamma_1^{-1}$ satisfies the 2nd point in Def.~\ref{def-bisim}, for all $q^*=(l,v^*) \in Q(\game_W) $ and $q=(l,v)=\gamma_1^{-1}(q^*) \in Q(\game_S)$. Assume that $q^*$ a configuration of $\playerone$, for all $m=(a,t) \in \mvs{T(\game_W)}$ and ${q^*}'=(l',{v^*}')=\delta_W(q^*,m)$ in $T(\game_W)$, we take the same move $m$ in $T(\game_S)$, we first prove that $\delta_S(q,m)$ is well defined in $T(\game_S)$, let us call it $q'$. Afterwards we prove that $q'=\gamma_1^{-1}({q^*}')$. We have for each $x\in \var$,
$$ {v^*}'(x) ={v^*}(x) +t \cdot \flow_W(l,x) \in \varphi_{e^*}(x)$$
And ${v^*}'=\post_{e^*}(v^*)$ for the edge $e^*=(l,a,\varphi_{e^*}(x), \post_{e^*}(x),l')$, $e^*$ is obtained from the edge $e=(l,a,\varphi(x), \post(x),l')$ in $\game_S$ by transformation of Sec.~\ref{trans:singToStop}. Now we show that $v + \flow(l,\cdot) \in \varphi_e$. For all $x\in \var$ and $\flow(l,x)\neq 0$ we have
\begin{align*}
&{v^*}(x) +t \cdot \flow_W(l,x) \in \varphi_{e^*}(x)\\
&\implies {v^*}(x) \cdot \flow(l,x) +t \cdot \flow_W(l,x) \cdot \flow(l,x) \in \varphi_{e^*}(x)\cdot \flow(l,x)\\
&\implies v(x)+t \cdot \flow(l,x) \in \varphi_{e^*}(x)\cdot \flow(l,x)\\
&\implies v(x)+t \cdot \flow(l,x) \in \varphi_e(x)
\end{align*}
For all $x\in \var$ and $\flow(l,x) = 0$, we obtain $v^*(x)=v(x)$ and $\varphi_{e^*}(x)=\varphi_e(x)$. It follows that $v(x)\in \varphi_e(x)$. This implies that indeed $q'=\delta_S(q,m)=(l',v')$ with $v'=\post_e(v +t \cdot \flow(l,\cdot))$ is well defined. 

Now let us prove $q'=\gamma_1^{-1}({q^*}')$. For all $x\in\var$: \\
(1) For $\flow(l,x)\neq 0$ and $\post_{e^*}=\bot$ we have:
\begin{equation*}
\begin{split}
\gamma_1^{-1}(l', {v^*}')  &= (l', v^{*'}(x) \cdot {\flow(l',x)}
= (l', ( v^*(x) + t \cdot \flow_W(l,x) ) \cdot {\flow(l',x)})\\
&= (l', ( v^*(x) + t )  \cdot \flow(l,x) )
= (l', ( v^*(x)  \cdot \flow(l,x)  + t \cdot \flow(l,x)))\\
&= (l', v + t \cdot \flow(l,x))=(l',v'(x))
\end{split}
    \end{equation*}
(2) For $\post_{e^*}\neq \bot$ and $\flow(l',x)\neq 0$ we have:
\begin{equation*}
\begin{split}
\gamma_1^{-1}(l', {v^*}')  &= (l', v^{*'}(x) \cdot {\flow(l',x)}= (l', \post_{e^*}(x) \cdot {\flow(l',x)})\\
&= (l', \post_e(x) )=(l',v'(x))
\end{split}
    \end{equation*}
(3) For $\post_{e^*}\neq \bot$ and $\flow(l',x)=0$ we have:
\begin{equation*}
\begin{split}
\gamma_1^{-1}(l', {v^*}')  &= (l', v^{*'}(x))= (l', \post_{e^*}(x))= (l', \post_e(x) )=(l',v'(x))
\end{split}
    \end{equation*}
We conclude that $(l',v')=\gamma_1^{-1}(l', {v^*}')$.

Now, for proving the 3rd point in Def. \ref{def-bisim}: let $q^*$ be a configuration of $\playertwo$, and let $m$ be a move in $\mvs{\game_S}$ and $q'=(l',v')=\delta_S((l,v),m)$ in $T(\game_S)$, we take the same move in $\mvs{\game_W}$, we first prove that $q{^*}'=\delta_W((l,v^*),m)$ is well defined in $T(\game_W)$, and then we prove that $\gamma_1^{-1}(q{^*}')=(q')$. \\
Since $q'=(l',v')=\delta_S((l,v),m)$ then there exists an edge $e=(l,a,\varphi_{e},\post_{e},l')$. 
By transformation there exists an edge in $\game_W$ $e^*=(l,a,\varphi_{e^*},\post_{e^*},l')$. 

For each $x \in \var$ with $\flow(l,x) \neq 0$, we have:
\begin{align}
   v(x) + t \cdot \flow(l,x) & \in \varphi_e(x) \notag \\
   & \implies v^*(x) \cdot \flow(l,x) + t \cdot \flow(l,x) \in \frac{\varphi_e(x)}{\flow(l,x)} \cdot \flow(l,x) \notag \\
      &\implies v^*(x) + t \in \varphi_{e^*}(x) \notag \\
    &\implies v^*(x) + t \cdot \flow_W(l,x) \in \varphi_{e^*}(x) 
    \label{varphigammainverse}
\end{align}
Notice that for $\flow(l,x) = 0$, we obtain $\flow_W(l,x) = 0$, $v^*(x)=v(x)$, $\varphi_e(x)=\varphi_{e^*}(x)$, hence $v^*(x) \in \varphi_{e^*}(x)$. 
Let ${v^*}'= \post_{e^*}(v^* + t \cdot \flow_W(l,\cdot))$. We have proved that $q{^*}'= (l,{v^*}')=\delta_W((l,v^*),m)$ is well defined. 

Now we prove that $\gamma_1^{-1}(l',{v^*}')=(l',v')$. For each $x \in \var$ consider the case $\flow(l',x) \neq 0$ and $\post_{e^*}=\bot$,
\begin{align*}
\gamma_1^{-1}(l',v^{*'}(x))&= (l',v^{*'}(x) \cdot \flow(l',x)) 
= \big(l',( v^*(x) + t \cdot \flow_W(l,x) ) \cdot \flow(l',x) \big) \\
&=  \big(l',v^*(x) \cdot \flow(l',x) + t \cdot \flow(l',x)  \big) \\
& =  \big(l',v^*(x) \cdot \flow(l,x) + t \cdot \flow(l,x) \big)  \\
&=  \big(l',v(x)  + t \cdot \flow(l,x) \big) 
= (l',v'(x))
    \end{align*}
Now for each $x \in \var$ $\flow(l',x) = 0$ and $\post_{e^*}=\bot$,
\begin{align*}
\gamma_1^{-1}(l',v^{*'}(x))&= (l',v^{*'}(x)) = (l',v^*(x)) 
= (l',v(x)) 
= (l',v'(x))
\end{align*}
Same for the case for each $x \in \var$ with $\post_{e^*}(x)\neq\bot$ and $\flow(l',x) = 0$,
\begin{align*}
\gamma_1^{-1}(l',v^{*'}(x))&=(l',v^{*'}(x))=(l',\post_{e^*}(x)) = (l',\post_e(x)) = (l',v'(x))
    \end{align*} 
Finally consider the case for each $x \in \var$ with $\post_{e^*}(x)\neq\bot$ and $\flow(l',x) \neq 0$,
\begin{align*}
\gamma_1^{-1}(l',v^{*'}(x))&= (l',\post_{e^*}(x)\cdot \flow(l',x)) = \left( l', \frac{\post_e(x)}{\flow(l',x)}\cdot \flow(l',x) \right) \\
&=  \big(l',\post_e(x) \big) 
= (l',v'(x))
    \end{align*}
This ends the proof.\qed
\end{proof}
\section{From Initialized Stopwatch to Updatable Timed Games}
The next step is to transform each initialized stopwatch game into a game where the dynamics of each variable is never zero, by eliminating all the flows of value zero. The games obtained by this transformation are called \emph{updatable timed games}. 
\begin{definition}[Updatable Timed Game]
An updatable timed game $\game_{U}$ of $n$-dimension is an initialized singular game in 
which all variables are clocks, i.e. $\flow(l,x) = 1$ for any location $l$ and any 
variable $x\in \var$.
\end{definition}
Note that these structures generalize timed games by allowing updates with any rational value, similarly with updatable timed automata from \cite{BouyerDFP04}. 
\subsection{Transformation}
Given $\game_W$ an initialized stopwatch game with $n$ variables (whose semantics is denoted $T(\game_W)$ in the sequel) we construct an updatable timed game in two steps.
We first transform $\game_W$ into another stopwatch game $\overline{\game_W}$ 
in which the locations carry some extra information about resets.
Then this stopwatch game is transformed into an updatable timed game.

The need for the intermediate stopwatch game comes from the fact that each stopwatch in the original game is simulated by a clock in the resulting game.  But clocks 
are always incremented when time pass, so when some stopwatch in a location $l$ in which the stopwatch's flow is $0$ must be encoded with a clock, 
we need to reset this clock on any transition leaving $l$. The value to which this clock must be updated depends on the sequence of transitions through which $l$ has been reached.
Along a finite run $\rho$ ending with the edge $e$ this reset value it is the update to which the stopwatch was reset on the latest edge $e'$ when the stopwatch's flow changed from $1$ to $0$. 
These bits of memory are modeled by the extra information encoded into locations of the new stopwatch game.

The formalization of the first transformation is the following: 
$$\overline{\game_W}= (\overline{l_0}, \overline{L_1}, \overline{L_2},\var, \act, \obs,\flow, \overline{\edge},\labeling)$$
\begin{enumerate}
    \item $\overline{l_0}=(l_0,f_{l_0})$ where $f_{l_0}: \var \to \{0\}$.
    \item $\overline{L_1}= L_1 \times F$ where $F$ is the set that consists of 
    functions $f : \var \to K_{\bot}$, where $K$ 
    is the set of all constants used in $\game_W$ and 
    $K_{\bot}=K \cup \{\bot\}$.
    \item $\overline{L_2}=L_2 \times F$. 
    \item $\labeling$ is extended over $\overline{L_1} \cup \overline{L_2}$ as expected.
    \item $\overline{\edge}$ is the set of edges $\overline{e}=(\overline{l_1},a,\varphi, \post,\overline{l_2})$ where 
    $\overline{l_1} =(l_1,f_{l_1}), \overline{l_2} = (l_2,f_{l_2})$ and
    there exists an edge $e\in \edge$ with $e=(l_1,a,\varphi,\post,l_2)$ 
    and for all $x \in \var$:
    \begin{equation}
    \label{transformation-2}
        f_{l_2}(x) = 
        \begin{cases}
            \bot & \text{ if } \flow(l_2,x)=1\enspace,\\
            \post(x) & \text{ if } \flow(l_2,x) = 0 \land \post(x)\neq \bot\enspace,\\
            f_{l_1}(x) & \text{ if } \flow(l_2,x) = 0 \land \post(x)=\bot\enspace.
        \end{cases}
    \end{equation}
\end{enumerate}
The semantics of $\overline{\game_W}$ that we call $T(\overline{\game_W})$ is as follows
$$T(\overline{\game_W}) = (Q(\overline{\game_W}), \mvs{\overline{\game_W}}, \overline{q_0},\overline{\delta_W})$$
Now we construct $\game_U$ an updatable timed game from $\overline{\game_W}$ in the next step.\\
Step 2:
$$\game_{U} = (\overline{l_0}, \overline{L_1}, \overline{L_2}, \var, \act, \obs, \edge_U, \labeling)$$
\begin{enumerate}
    \item $\overline{l_0}, \overline{L_1}, \overline{L_2}, \obs, \act,\labeling$ are all same as in $\overline{\game_W}$.
    \item $\edge_U$ is the set of edges $e_U=(\overline{l_1},a,\varphi, \post_U,\overline{l_2})$ with $\post_U: \var\to \mathbb{Q}^n \cup \{\bot\}$ such that for all $\overline{e} \in \overline{\edge}$ with $\overline{e}=(\overline{l_1}, a, \varphi, \post,\overline{l_2})$ we have for all $x \in \var$,
    \begin{equation}
       \label{overline.post}
       \post_U(x)=
     \begin{cases}
        f_l(x) & \text{ if }f_l(x) \neq \bot\enspace,\\
        \post(x) & \text{ otherwise }\enspace.     
      \end{cases}
    \end{equation}
        
Note that if we have a reset of a variable $x$ with $flow(l,x)=0$, then $\post_U(x)=\post(x)=f_l(x)$. Note that $\edge_U$ is same as $\overline{\edge}$ with different reset function $\post_U(\var)$ on each edge.
\end{enumerate}
The transition system of the updatable timed game $\game_U$, that we call $T(\game_U)$ is as follows
$$T(\game_U)= (Q(\game_U),\mvs{\game_U}, q_{0_U},\delta_U)$$
We define the relation 
$\beta_1 \subseteq Q(\game_W) \times Q(\overline{\game_W})$ as follows: 
\[
\beta_1 = \{ ((l,v), (\overline{l},\overline{v})) \mid (l,v) \in Q(\game_W), \overline{l}=(l,f_l) \text{ where } f_l : X \to K_{\bot} \text{ and } \overline{v} = v \}
\]
And its inverse $\beta_1^{-1} \subseteq Q(\overline{\game_W}) \times Q(\game_W)$,
\[
\beta_1^{-1}=\{ \left( (\overline{l},\overline{v}),(l,v)\right)  \mid (\overline{l},\overline{v}) \in Q(\overline{\game_W}), ((l,v),(\overline{l},\overline{v}))\in \beta_1\}
\]
\subsection{Bisimulation between $T(\game_W)$ and $T(\overline{\game_W})$}
In this section we prove that $\beta_1$ is a bisimulation relation.
\begin{lemma} \label{first}
$\beta_1$ witnesses $T(\game_W) \preceq_s T(\overline{\game_W})$.
\end{lemma}
\begin{proof} \label{Proofbeta1}
Note that, already by definition $(\overline{l_0},v_0)=\left( (l_0,f_{l_0}) ,v_0 \right)$ and $(l_0,v_0)$ are the initial configurations in the two games, hence $\left((l_0,v_0),(\overline{l_0},v_0)\right) \in \beta_1$ as required. 

Let $(l,v) \in Q({\game_W})$ and $ (\overline{l},v)=\left((l,f_l),v\right) \in Q(\overline{\game_W})$ with $\left((l,v), (\overline{l},v)\right) \in \beta_1$.
Since the location $l$ is the same, the observation is preserved thus $\labeling(l,v)=\labeling (\overline{l},v)$. So the first point of the Def. \ref{def-bisim} holds. 

Let $(l,v)$ be a configuration of $\playerone$, let $m=(a,t) \in \mvs{\game_W}$ and $(l',v') = \delta_W((l,v),(a,t))$ in $T(\game_W)$.
To prove point 2 in Def. \ref{def-bisim}, we take the same $m$ in $\mvs{\overline{\game_{W}}}$ and prove that $\overline{\delta_W}((\overline{l},v),(a,t))=(\overline{l'},v')$ 
is well defined, and then $\left( (l',v'),(\overline{l'},v')\right) \in  \beta_1 $.
\label{SENS1}
$\delta_W((l,v),(a,t))=(l',v')$ implies that there exists an edge $e= (l,a,\varphi_e,\post_e,l')$ in $\game_W$, hence, by construction \ref{transformation-2} there exists $\overline{e}=((l,f_l), a, \varphi_e, \post_e,(l',f_{l'}))$ in $\overline{\game_W}$. 
This implies that $\overline{\delta_W}((\overline{l},v),(a,t))=(\overline{l'},v')$ with 
$\overline{l'}=(l',f_{l'})$ is well defined in $T(\overline{\game_W})$. It follows that $\left( (l',v'),(\overline{l'},v')\right) \in  \beta_1 $ since $v'$ is the same for the two configurations and $f_{l'}$ is well defined.

For proving the third point of the Def. \ref{def-bisim}, let $(l,v)$ be a configuration of $\playertwo$,
and $m=(a,t) \in \mvs{\overline{\game_{W}}}$ with $\overline{\delta}\left(((l,f_l),v),(a,t)\right)=((l',f_{l'}),v')$ in $\overline{\game_W}$. 
Similarly to the above proof, we take the same move $m$ in $ \mvs{\game_{W}}$, then prove that $\delta_W((l,v),(a,t))=(l',v')$ is well defined and $\left( (l',v'),(\overline{l'},v')\right) \in  \beta_1 $.
\label{SENS2}
This follows by noting that the existence of transition $\overline{\delta}\left(((l,f_l),v),(a,t)\right)=((l',f_{l'}),v')$ implies there exists an edge $\overline{e}= ((l,f_l),a,\varphi_e,\post_e,(l',f_{l'}))$ in $\overline{\game_W}$. 
By construction \ref{transformation-2}, $\overline{e}$ is obtained from $e= (l,a,\varphi_e,\post_e,l')$ in $\game_W$. 
Note that $\post_e$ is preserved by the construction. Hence $(l',v')=\delta_W((l,v),(a,t))$ with $v'=\post_e(v+t\cdot \flow(l,\cdot))$ is well defined, which ends the proof. 
\qed\end{proof}
 \begin{lemma}
 $\beta_1^{-1}$ witnesses $ T(\overline{\game_W}) \preceq_s T(\game_W)$.
 \end{lemma}
 The proof of this lemma is similar to the proof of Lemma \ref{first}

We define the mapping
$\beta_2: Q(\overline{\game_W}) \to Q({\game_U})$ as follows: 
\[\forall (\overline{l},v)=((l,f_l),v) \in Q(\overline{\game_W}),~\beta_2 \colon ((l,f_l),v) \mapsto ((l,f_l),v) \enspace.\]

\begin{lemma} \label{bijec}
$\beta_2$ is the bisimulation between $T(\overline{\game_W})$ and $T(\game_U)$.
\end{lemma}
\begin{proof}
$\beta_2$ is the identity function hence it is clear that it is a bisimulation.
\qed\end{proof}

Now let $\beta: Q(\game_W) \to Q(\game_U)$ as $\beta= \beta_1 \circ \beta_2$. And it's inverse $\beta^{-1}: Q(\game_U) \to Q(\game_W)$ as $\beta^{-1}= \beta_2^{-1} \circ \beta_1^{-1}$.
Lemmas \ref{bijec} and \ref{compo} directly imply the following:
\begin{corollary} \label{cor1}
$\beta$ is a bisimulation between $T(\game_W)$ and $T(\game_U)$.
\end{corollary}

\section{From Updatable Timed Game to Timed Game}
\begin{definition} [Timed Turn Based Game]
A timed turn based game $\game_{T}$ of $n$-dimension is an updatable timed game in which clocks can only be reset to 0.
\end{definition}
In the sequel, in an edge $e=(l,a,\varphi_e,r,l')$ of a timed automaton, we consider that $r \subseteq X$ denotes the set of clocks which are reset along this edge.  

\subsection{Transformation}
\label{updatabletotimed}
Given an updatable game $\game_U= (l_0,L_1, L_2, \var, \act, \obs, \edge, \labeling)$ whose semantics is denoted $T(\game_U)= (Q(\game_U),\mvs{\game_U}, q_0, \delta_U)$ we now construct the timed game $\game_T$ as:
\[
\game_T= (l_0^t, L_1^t, L_2^t, \var, \act, \obs,\edge^t,\labeling)
\]
where 
\begin{enumerate}
\item  $l_0^t=(l_0,g_{l_0})$ where $g_{l_0}: \var \to \{0\}$.
    \item $L_1^t= L_1 \times F_t $ where $F_t$ is the set of functions $g: \var \to K_U$, where $K_U$ is the set of constants used in the definition of $\game_U$.
    \item $L_2^t=L_2 \times F_t$. 
    \item $\labeling$ is extended over $L_1^t \cup L_2^t$ as expected.
    \item $\edge^t$ is the set of edges $e^t=(l^t_1,a,\varphi_{e^t},r,l^t_2)$ where $l_1^t=(l_1, g_{l_1})$ and $l_2^t=(l_2, g_{l_2})$ such that there exists an edge $e\in \edge$ with $e=(l_1,a,\varphi_e,\post_e,l_2)$ and for all $x \in \var$:
\begin{align}
        \label{varphi.t}
    \varphi_{e^t} & = \bigwedge_{x\in \var} \varphi_e(x) \text{ with  } \varphi_{e^t}(x) = \varphi(x)- g_{l_1}(x)    \\
    \label{g.t}
     g_{l_2}(x) & = 
        \begin{cases}
            \post(x) & \text{ if } \post(x)\neq \bot \land x\in r \\
            g_{l_1}(x) & \text{ if } \post(x)=\bot
        \end{cases} \\
        \label{r}
     x \in r & \iff \post_e(x)\neq \bot
\end{align}
\end{enumerate}
The  transition system of the constructed timed game $\game_T$ is denoted:
$$T(\game_T)=(Q(\game_T),\mvs{\game_T},q_0^t,\delta_T) $$

We define the relation 
$\gamma_2 \subseteq Q(\game_U) \times Q(\game_T)$ as follows: 
\begin{align*}
\gamma_2 = \{ ((l,v),(l^t,v^t)) \mid & (l,v) \in Q(\game_U),~l^t=(l,g_l) \text{ where }  g_l : X \to K_U   \\
& \text{ and for all } x\in\var~v^t(x)=v(x) - g_l(x) \}
\end{align*}
And its inverse $\gamma_2^{-1} \subseteq Q(\game_T) \times Q(\game_U)$,
\begin{align*}
\gamma_2^{-1}  =\{ \left( (l^t,v^t),(l,v)\right)  \mid & (l^t,v^t)=((l,g_l),v^t) \in Q(\game_T),  \\
 &\forall~x\in\var~ v(x)=v^t(x) + g_l(x) \}
\end{align*}
\subsection{Bisimulation between $T(\game_U)$ and $T(\game_T)$}
We now show that the mapping $\gamma_2$ defined above is an alternating bisimulation.
\begin{lemma}
\label{lemma.gamma2}
   $\gamma_2$ witnesses that $T(\game_U) \preceq_s T(\game_T)$.
\end{lemma}

\begin{proof}
Recall that $(l_0^t,v^t_0)=((l_0,g_{l_0}),v_0^t)$ is the initial configuration in $T(\game_T)$ and $(l_0,v_0)$ is the initial configuration in $T(\game_U)$, where $v_0=v_0^t$, hence $((l_0,v_0),(l_0^t,v_0^t))\in \gamma_2$ as required. 

Let $(l,v) \in Q(\game_U)$ and $(l^t,v^t)=((l,g_{l}),v^t) \in Q(\game_T)$ with $((l,v),(l^t,v^t))\in \gamma_2$.
Since the location $l$ is the same, the observation is preserved thus $\labeling(l,v) = \labeling(l^t,v^t)$, so the first point in Def.~\ref{def-bisim} holds.

We now show that $\gamma_2$ satisfies the 2nd point in Def.~\ref{def-bisim}. 
Let $(l,v) \in Q_1(\game_U)$, $m=(a,t) \in \mvs{\game_U}$ and $\delta_U((l,v),m)=(l',v')$.
We then apply the same move in $T(\game_T)$ and prove that $\delta_T((l^t,v^t),m)=({l^t}',{v^t}')$ is well defined, and $\left((l',v'),({l^t}',{v^t}')\right) \in \gamma_2$.

Note that $\delta_U((l,v),m)=(l',v')$ implies that there exists an edge $e=(l,a,\varphi_e,\post_e,l')$ in $\game_U$. We obtain by transformation of Sect.~\ref{updatabletotimed} the  
edge $e^t =((l,g_{l}), a, \varphi_{e^t}(\var),r,(l',g_{l'}))$ in $\game_T$. For all $x\in \var$ we have:
\begin{align*}
&v (x)+ t  \in \varphi_{e}(x)\\
&\implies v (x) + t - g_l (x)  \in \varphi_{e}(x) - g_l(x) \\
&\implies v(x) - g_l(x) +t \in \varphi_{e^t}(x) \text{ cf.~ the Identity \eqref{varphi.t}.}\\
&\implies v^t(x) +t  \in \varphi_{e^t}(x) \text{ because $((l,v),(l^t,v^t))\in \gamma_2$.}
\end{align*}
We obtain $\delta_T(((l,g_{l}),v^t),m)=((l',g_{l'}),{v^t}')$, where ${v^t}'= v^t + t [r:=0]$ is indeed well defined in $T(\game_T)$.
It remains to show that $\left((l',v'),\left((l',g_{l'}),{v^t}'\right)\right)$  are in $\gamma_2$. 
For each $x \in \var$, consider the case $x \notin r$:
        \begin{align*}
            {v^t}'(x) = v^t (x)+t  &= v(x) - g_{l}(x)+ t \\
            &= v(x) +t - g_{l'}(x) \text{ since, by construction, $g_{l}(x)=g_{l'}(x)$ for $x \notin r$.} \\
            &= v'(x)-g_{l'}(x) \text{ because $\post(x)=\bot$ cf. Identity~\eqref{r}}
                \end{align*}
 Now consider $x \in r$:       
                  \begin{align*}
           {v^t}'(x) = 0 &= \post_e(x) - g_{l'}(x) \text{ since, by construction,  $g_{l'}(x)= \post_e(x)$ for $x\in r$ }\\
            &= v'(x)- g_{l'}(x) \text{ since $\post_e(x)\neq \bot$.}
                \end{align*}
We conclude that $\left((l',v'),((l',g_{l'}),{v^t}')\right) \in \gamma_2$.
            
To show that $\gamma_2$ satisfies the third point in Def.~\ref{def-bisim}, let $(l,v) \in Q_2(\game_U)$, 
let $m=(a,t)\in \mvs{\game_T}$ and $\delta_T(((l,g_{l}),v^t),m)=((l',g_{l'}),{v^t}')$. We again apply the same move in $T(\game_U)$ and
prove, first, that $\delta_U((l,v),m)=(l',v')$ is well defined, and then that $\left((l',v'),((l',g_{l'}),{v^t}')\right) \in \gamma_2$.

Note that $\delta_T(((l,g_{l}),v^t),m)=((l',g_{l'}),{v^t}')$ implies that there exists an edge $e^t=((l,g_{l}),a,\varphi_{e^t},r,(l',g_{l'}))$ in $\game_T$. 
$e^t$ is constructed by the transformation in Subsection \ref{updatabletotimed} from the edge $e =(l, a, \varphi_{e},\post_e,l')$ in $\game_U$. For all $x\in \var$ we have:
\begin{align*}
&v^t(x)+ t  \in \varphi_{e^t}(x)\\
&\implies v(x) - g_l(x) + t \in \varphi_{e^t}(x) \text{ because $((l,v),(l^t,v^t))\in \gamma_2$.} \\
&\implies v(x) - g_l(x) + t \in \varphi_{e}(x)- g_l(x) \text{ cf. Identity \ref{varphi.t}.} \\
&\implies v(x) +t  \in \varphi_{e}(x) 
\end{align*}
Hence $\delta_U((l,v),m)=(l',v')$ with $v'(x)=\post_e(v+t)(x)$ for all $x\in \var$ is indeed well defined in $T(\game_U)$. 

Now we prove that $\left((l',v'),((l',g_{l'}),{v^t}')\right) \!\in \!\gamma_2$. For $x\!\in\! \var$ consider the case $x \notin r$:
\begin{align*}
          v'(x) - g_{l'}(x) &=  v(x) +t - g_{l'}(x) \text{ because $x \notin r$ implies $\post(x) = \bot$ by identity \ref{r}.} \\
                &=  v(x) +t - g_{l}(x) \text{ by identity \ref{g.t}, $\post(x) = \bot \implies g_l(x)=g_{l'}(x)$.}\\
                 &=  v(x)^t+t \text{ because $\left((l,v),((l,g_{l}),{v^t}\right) \in \gamma_2$.}\\
                  &={v^t}'(x)
\end{align*}
On the other hand, for the case $x \in r$:
\begin{align*}
 v'(x) - g_{l'}(x) &=  \post_e(x) - g_{l'}(x) 
 \text{ 
  because $\post(x) \neq \bot$.} \\
                &= 0 \text{ because $g_{l'}(x)=\post_e(x)$}\\
                &= {v^t}'(x) \text{ because $x\in r$.}
\end{align*}
This concludes the proof.\qed
\end{proof}
\begin{lemma}
 $\gamma_2^{-1}$  witnesses that $T(\game_T) \preceq_s T(\game_U)$.
\end{lemma}

\begin{proof}
Recall that $(l_0^t,v^t_0)=((l_0,g_{l_0}),v_0^t)$ is the initial configuration in $T(\game_T)$ and $(l_0,v_0)$ is the initial configuration in $T(\game_U)$, where $v_0^t=v_0$ because $g_{l_0}=0$, hence $((l_0^t,v_0^t),(l_0,v_0))\in \gamma_2^{-1}$ as required. 

Let $(l^t,v^t)=((l,g_{l}),v^t) \in Q(\game_T)$ and $(l,v) \in Q(\game_U)$ with $((l^t,v^t),(l,v))\in \gamma_2^{-1}$.
Since the location $l$ is the same, the observation is preserved thus $\labeling(l^t,v^t)=\labeling(l,v)$, therefore the first point in Def.~\ref{def-bisim} holds.

We now show that $\gamma_2^{-1}$ satisfies the second point in Def.~\ref{def-bisim}. Let $(l^t,v^t) \in Q_1(\game_T)$, a configuration of $\playerone$,
let $m=(a,t)\in \mvs{\game_T}$ and $\delta_T(((l,g_{l}),v^t),m)=((l',g_{l'}),{v^t}')$. By applying the same move in $T(\game_U)$, we prove in the first place $\delta_U((l,v),m)=(l',v')$ is well defined, and later on we prove $\left(((l',g_{l'}),{v^t}'),(l',v')\right) \in \gamma_2^{-1}$.

Note that $\delta_T(((l,g_{l}),v^t),m)=((l',g_{l'}),{v^t}')$ implies that there exists an edge $e^t=((l,g_{l}),a,\varphi_{e^t},r,(l',g_{l'}))$ in $\game_T$. $e^t$ is constructed by the transformation of the Sect. \ref{updatabletotimed} from the edge $e =(l, a, \varphi_{e},\post_e,l')$ in $\game_U$. For all $x\in \var$ we have:
\begin{align*}
& v(x)+ t  \in \varphi_{e}(x)\\
&\implies v^t(x) + g_l(x) + t \in \varphi_{e}(x) \text{ because $\left(\left(l^t,v^t\right),(l,v)\right)\in \gamma_2^{-1}$.} \\
&\implies {v^t}'(x) + g_l(x) \in \varphi_{e}(x) \\
&\implies {v^t}'(x) \in \varphi_{e}(x) - g_l(x)  \\
&\implies {v^t}'(x)  \in \varphi_{e^t}(x) \text{ cf.~ Identity \eqref{varphi.t}.}
\end{align*}

Hence $\delta_U((l,v),m)=(l',v')$ with $v'(x)=\post_e(v+t)(x)$ for all $x\in \var$ is well defined in $T(\game_U)$. 

To prove that $\left(\left((l',g_{l'}),{v^t}'\right),(l',v')\right) \in \gamma_2^{-1}$, for $x\in \var$ with $x \notin r$ we have:
\begin{align*}
          {v^t}'(x) + g_{l'}(x) &=  {v^t}'(x) + g_{l}(x)  \intertext{ since $x\notin r$ implies $\post(x) = \bot$ by Identity \eqref{r}, hence $g_l(x)=g_{l'}(x)$.}
         & =  v^t(x) +t + g_{l}(x)  =  v(x)+t \text{ because $\left(((l,g_{l}),{v^t}),(l',v')\right) \in \gamma_2^{-1}$.}\\
          &=v'(x)      
\end{align*}
Now consider the case $x \in r$:
\begin{align*}
          {v^t}'(x) + g_{l'}(x) &=  g_{l'}(x)  \text{ because for all $x\in r,~{v^t}'(x)=0$}\\
          &=  \post_e(x) \text{ cf.~ transformation of Sect.~\ref{updatabletotimed}.} \\ 
          &=v'(x) \text{ because $\post_e(x) \neq \bot$ by Identity \eqref{r}.}     
\end{align*}
We conclude that $\left(\left((l',g_{l'}),{v^t}'\right),(l',v')\right) \in \gamma_2^{-1}$.

We now show that $\gamma_2^{-1}$ satisfies the third point in Def.~\ref{def-bisim}. 
Let $\left((l,g_{l}),v^t\right) \in Q_2(\game_T)$, a configuration of \playertwo, let $m=(a,t)\in \mvs{T(\game_U)}$ and $\delta_U((l,v),m)=(l',v')$. 
Again we take the same move $m$ in $\mvs{T(\game_T)}$ and show that $\delta_T(((l,g_{l}),v^t),m)=((l',g_{l'}),{v^t}')$ is well defined and that $\left(\left((l',g_{l'}),{v^t}'\right),(l',v')\right) \in \gamma_2^{-1}$.

Note that $\delta_U((l,v),m)=(l',v')$ implies the edge $e=(l,a,\varphi_e,\post_e,l')$ must exist in $\game_U$. We obtain by transformation \ref{updatabletotimed} the 
edge $e^t =((l,g_{l}), a, \varphi_{e^t}(\var),r,(l',g_{l'}))$ in $\game_T$. For all $x\in \var$ we have:
\begin{align*}
v (x)+ t  \in \varphi_{e}(x)
&\implies v^t (x) + t + g_l (x)  \in \varphi_{e}(x) \text{ because $((l^t,v^t),(l,v))\in \gamma_2^{-1}$.} \\
&\implies v^t(x) +t \in \varphi_{e}(x) - g_l(x)\\ 
&\implies v^t(x) +t  \in \varphi_{e^t}(x) \text{ cf.~ Identity \eqref{varphi.t}.}
\end{align*}
We obtain $\delta_T(((l,g_{l}),v^t),m)=((l',g_{l'}),{v^t}')$, where ${v^t}'= v^t + t [r:=0]$ is indeed well defined in $T(\game_T)$.
Let's show that $\left(\left((l',g_{l'}),{v^t}'\right),(l',v')\right) \in \gamma_2$. For all $x \in \var$, consider the case $x \notin r$:
\begin{align*}
        v'(x)&=v(x)+t \text{ because $x \notin r \implies \post_e(x)\neq\bot$ }  \\
        &= v^t(x)+g_l(x)+t \text{ because $((l^t,v^t),(l,v))\in \gamma_2^{-1}$.}\\
        &={v^t}'(x)+g_l(x) \text{ because $x \notin r$.} \\
        &={v^t}'(x)+g_{l'}(x) \text{ by the construction of Sect.~\ref{updatabletotimed} $g_{l}(x)=g_{l'}(x)$ for $x \notin r$.} 
        \end{align*}
 Now consider $x \in r$:       
                  \begin{align*}
          v'(x)= \post_e(x)&= g_{l'}(x) \text{ since, by construction, $g_{l'}(x)= \post_e(x)$ for $x\in r$ }\\
            &= {v^t}'(x)+ g_{l'}(x) \text{ because ${v^t}'(x)=0$ for $x \in r$.}
                \end{align*}

This concludes the proof. 
\qed\end{proof}
\section{Conclusion}
We have shown that each initialized singular game is bisimilar with a timed game.
We recall the different constructions in the figure below.
The  decidability of the synthesis problem for LTL specifications and initialized singular games can then be seen as a corollary of this result.
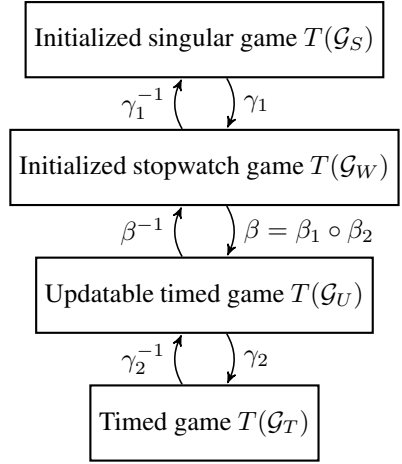
\begin{center}
\begin{tikzpicture}[->,>=stealth',shorten >=0.5pt,auto,node distance=1.7cm,semithick]

\tikzstyle{place}=[rectangle,thick,draw=black,fill=white,minimum size=10mm,text=black]
  \tikzstyle{every state}=[circle,thick,draw=red!75,fill=red!20,minimum size=20mm,]

\node[place]         (B)                {Initialized singular game $T(\game_S)$
                                                };
                                        
 \node[place]         (F)     [ below of=B] {Initialized stopwatch game $T(\game_W)$
                                                };
 \node[place]         (G)     [ below of=F] {Updatable timed game $T(\game_U)$
                                                
                                                };

\node[place]         (H)     [ below of=G] {Timed game $T(\game_T)$
                                                
                                                };
  \path 
   (B) edge   [bend left]     node { $\gamma_1$} (F)
   (F) edge  [bend left]  node { $\gamma_1^{-1}$}     (B)
(F) edge    [bend left]     node {$\beta= \beta_1 \circ \beta_2$ } (G)
(G) edge   [bend left]   node {$\beta^{-1}$ }   (F)
(G) edge   [bend left]      node {$\gamma_2$ } (H)
(H) edge   [bend left]   node {$\gamma_2^{-1}$ }   (G)
         ;   
\end{tikzpicture}
\end{center}

We plan to extend this result by carefully adapting the results for larger subclasses of initialized rectangular hybrid automata from \cite{HENZINGER199894}
and hence weakening the constraint that the derivative of each variable must be piecewise constant. 

\bibliographystyle{splncs04}
\bibliography{ref}

\end{document}